\documentclass[letterpaper, 10 pt, conference]{ieeeconf}  

\IEEEoverridecommandlockouts                              

\usepackage{graphics} 
\usepackage{epsfig} 
\usepackage{mathptmx} 
\usepackage{times} 

\usepackage{amsmath} 
\usepackage{amssymb}  
\usepackage{color}
\usepackage{graphicx}
\usepackage{float}

\newtheorem{theorem}{Theorem}

\title{
PESTO: Formally Correct Registration of LiDAR Point Clouds\\with Limited Overlap
}

\author{Valen Yamamoto$^{1}$, Matteo Marchi$^{1}$, and  Paulo Tabuada$^{1}$
\thanks{*This material is based upon work supported by the National Science Foundation Graduate Research Fellowship Program under Grant No. DGE-2444110 and DGE-2034835, and by NSF project 2211146. Any opinions, findings, and conclusions or recommendations expressed in this material are those of the author(s) and do not necessarily reflect the views of the National Science Foundation.}
\thanks{$^{1}$Valen Yamamoto, Matteo Marchi, and Paulo Tabuada are with the Electrical and Computer Engineering Department, University of California at Los Angeles, Los Angeles, CA 90095 USA (e-mail: {\tt\small \{valenyamamoto, matmarchi, tabuada\}@ucla.edu}).}%
}

\begin{document}

\maketitle
\thispagestyle{empty}
\pagestyle{empty}

\begin{abstract}
In this paper we tackle the problem of aligning LiDAR point clouds also known as the point cloud registration problem. We propose a new algorithm, PESTO, that exploits tetrahedra as ``universal features” for LiDAR data, i.e., features that are agnostic to the environment where the LiDAR sensors are deployed. We show empirically that PESTO is competitive with existing solutions for aligning LiDAR point clouds, especially in environments with occlusions. Moreover, we establish PESTO’s formal correctness by proving worst-case bounds on the alignment error. 
\end{abstract}

\section{Introduction}

With the increasing deployment of autonomous vehicles in traffic with human drivers, it is critical that algorithms within the self-driving autonomy stack provide performance and safety guarantees. Point cloud registration, the problem of finding the rigid transformation that aligns two point clouds, is an integral part of LiDAR localization algorithms. When these algorithms are used on real-world LiDAR data, we need to contend with several additional challenges.

First, the point clouds being aligned do not entirely describe the same area due to occlusions. In other words, the subset of the environment that is visible in both point clouds, that we term the \textit{common area}, is a strict subset of the point clouds. Implicitly or explicitly, each alignment algorithm identifies the common area, and this problem is combinatorial in nature as one must search over all possible subsets.

Second, the large number of points produced by 3D LiDAR sensors renders algorithms that work at the level of individual points too inefficient. To circumvent this, feature extraction methods were developed to distill large point clouds into a smaller sets of distinct features like corners and edges. However, these methods either require prior knowledge of the size and type of features prevalent in the environment or are data-driven deep learning solutions, both of which are not robust to changes in the environment.

Third, most point cloud registration algorithms fail silently, i.e., these algorithms produce correct and wildly incorrect estimates with the same confidence, unable to provide a measure of how well they have performed. With the addition of feature extraction methods, even point cloud registration algorithms with performance guarantees cannot provide an end-to-end bound on the quality of the produced rigid transformations. In particular, data-driven neural network solutions, often overfit to their training data, can fail unpredictably and provide a dramatically incorrect rigid transformation estimate for a pair of point clouds that are similar to a pair they correctly align; an example of this phenomenon can be seen in Figure \ref{fig:kitti_example}. As LiDAR point cloud registration is increasingly being deployed in safety-critical applications, it is crucial that alignment algorithms determine the uncertainty in their estimates.

Motivated by these difficulties, we propose PESTO\footnote[2]{\textbf{P}ASTA \textbf{E}mploying \textbf{S}implex \textbf{T}ransformation \textbf{O}bjective}, a point cloud registration algorithm with formal performance guarantees designed to work on LiDAR data with small common areas.

\begin{figure}[hbt!]
    \centering
    \includegraphics[width=\linewidth]{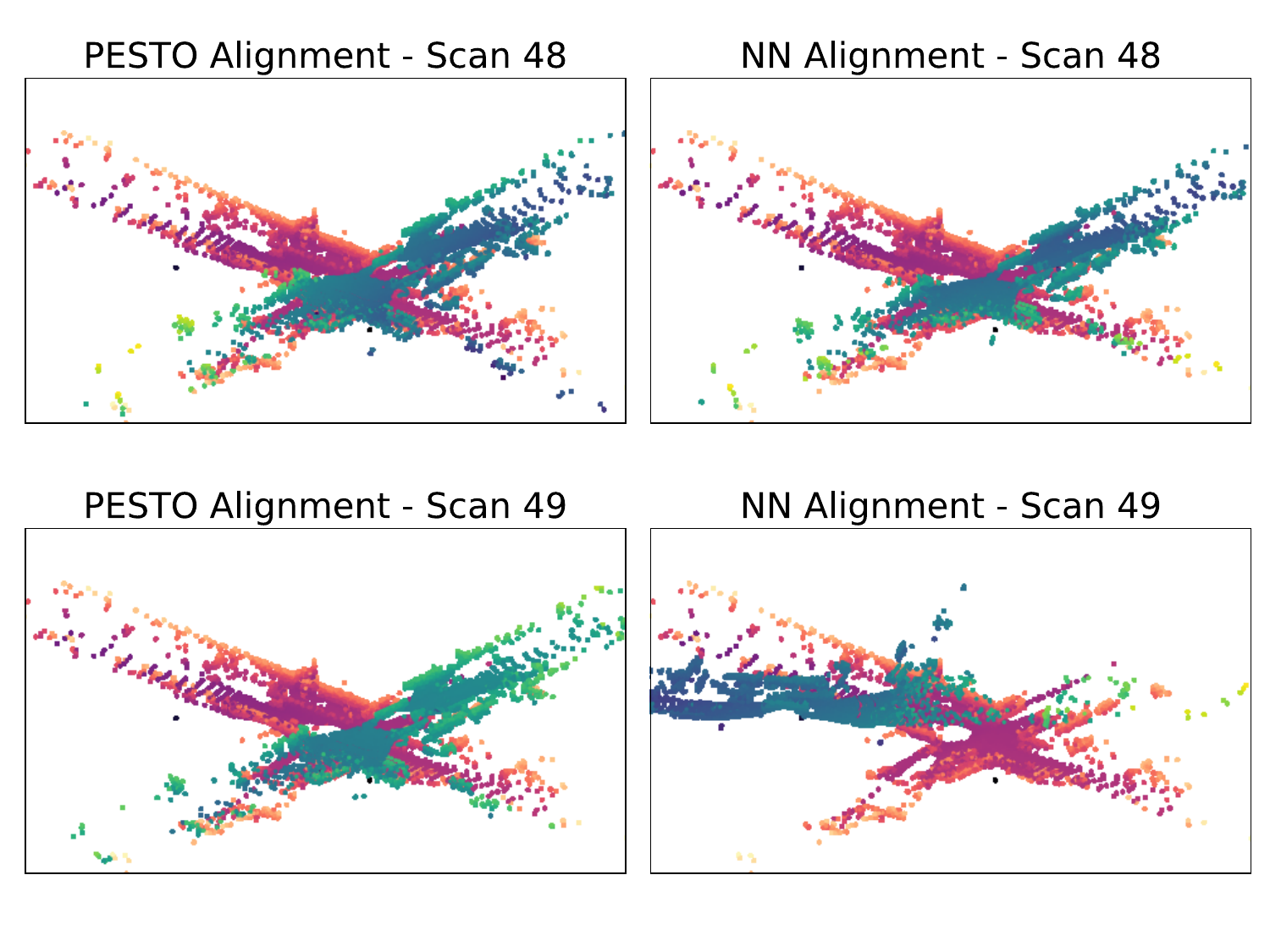}
    \caption{A visualization of alignments provided by PESTO and the Neural Network (NN) GeoTransformer for two pairs of LIDAR scans, one in red and the other in green, from the KITTI Odometry Benchmark, collected as the car makes a left turn at an intersection. Both tools correctly align scan 48, but GeoTransformer provides a very incorrect transformation for scan 49, even though it covers almost the same area as scan 48.}
    \label{fig:kitti_example}
\end{figure}

Our contributions are as follows:
\begin{enumerate}
    \item The PESTO algorithm, which uses Delaunay Tetrahedralization to efficiently compute the rigid transformation that best aligns the common area of two point clouds;
    \item Formal worst-case guarantees on the alignment error when using PESTO's rigid transformation estimate;
    \item An empirical comparison of PESTO to existing point cloud registration tools and a validation of PESTO's worst case guarantees on real world LiDAR data.
\end{enumerate}

\section{Related Work}
Point cloud registration algorithms fall into two categories: correspondence-based and correspondence-free. Correspondence-based algorithms take as input a set of corresponding features from both point clouds. Horn \cite{horn} provided a closed-form solution for the best, in a least-squares sense, rigid transformation between two sets of points. This method, however, is sensitive to incorrect (outlier) correspondences; to combat such correspondences, robust correspondence-based registration algorithms like RANSAC \cite{ransac} and other maximization consensus variants \cite{consensus} first find a set of correct (inlier) correspondences and then calculate their transformation estimates based on the correct correspondences. Correspondences are obtained using tools such as Fast Point Feature Histograms (FPFH) \cite{fpfh} or probabilistic methods like the normal distribution transform \cite{ndt} to find corresponding features. Correspondence-free algorithms simultaneously solve for the correspondences and the transformation. ICP \cite{icp} and its variants minimize a cost function that works directly with the points of the point cloud. Since they employ nonlinear optimization methods, they are only guaranteed to find local solutions close to an initial guess. Algorithms that are guaranteed to find the globally optimal solution, like Go-ICP \cite{goicp}, are too computationally expensive to run in real time. To run either correspondence-based or correspondence-free algorithms under real-time constraints, methods to downsample or distill the raw LiDAR point clouds into a smaller point cloud are needed to reduce amount of computations \cite{downsample}.

With the advent of neural network architectures which work directly on point clouds, deep learning-based tools for point cloud registration have become increasingly popular and have shown great accuracy on many benchmarks. In particular, the tools Predator \cite{predator} and GeoTransformer \cite{geotransformer} have shown great accuracy in aligning point clouds with low overlap; however, these algorithms suffer from the same drawbacks that all deep learning perception methods are subject to. They require very large amounts of data to train and do not generalize to different environments and LiDAR models, requiring users to collect and label a glut of data to re-train the network for every new environment they expect to use the neural network in \cite{deep_learning_survey}. While neural networks outperform model-based methods on benchmarks, they tend to fail in a sporadic, uninterpretable manner, making them an undesirable option for safety-critical systems.

Few point cloud registration tools provide performance guarantees. Tools that return the globally optimal transformation such as Go-ICP, some maximization consensus algorithms, and CGA \cite{cga}, have guarantees on returning the transformation that minimizes their cost functions, but their minima may not correspond to the estimate with the correct transformation. TEASER \cite{teaser} and PASTA \cite{pasta} provide worst-case performance guarantees on their estimates but the quality of these guarantees is conditional on assumptions that limit their use on real-world data. 

\section{Preliminaries}
Denote the set of real numbers by $\mathbb{R}$, the set of nonnegative real numbers by $\mathbb{R}_{\ge0}$, the set of strictly positive real numbers by $\mathbb{R}_{>0}$, and the set of natural numbers starting at $1$ by $\mathbb{N}$.

Given a set $P$, denote its boundary by $\partial P$, its volume by $\text{Vol } P$, and its  cardinality by $\vert P\vert$.  We use $B_\epsilon(x) = \left\{y \in \mathbb{R}^n \vert \Vert x-y \Vert < \epsilon\right\}$ to denote the ball of radius $\epsilon \in \mathbb{R}_{>0}$ centered at $x\in \mathbb{R}^n$. 
The first moment of a finite set $P \subset \mathbb{R}^n$ is given by $\mu(P)=\frac{1}{\vert P\vert }\sum_{x\in P} x$, and its second moment by $\Sigma(P)=\frac{1}{\vert P \vert }\sum_{x\in P} \left(x-\mu(P)\right)\left(x-\mu(P)\right)^\intercal$ where $(\cdot)^\intercal$ denotes matrix transposition.

The symbol $I$ represents the identity matrix. A matrix $M\in\mathbb{R}^{n\times n}$ is orthogonal if $M^\intercal  M=I$. The set of all such $n\times n$ matrices forms a group under matrix multiplication, called the orthogonal group, which we denote by $O(n)$. If we add the additional constraint that the determinant of $M$ must be $1$ we have the special orthogonal group, denoted by $SO(n)$.

A rigid transformation $(R, \rho)$ consists of a rotation matrix $R \in SO(n)$ and a translation vector $\rho \in \mathbb{R}^n$. A vector $x\in \mathbb{R}^n$ is transformed by $(R, \rho)$ into the vector $Rx + \rho\in \mathbb{R}^n$.

\section{Problem Statement}

In this paper, we consider the problem of aligning point clouds generated by a 3D LiDAR sensor. A LiDAR sensor's measurement, called a scan, consists of a set of distance measurements from where the sensor is placed in the environment to various other objects around it. Hence, scans depend on the sensor's location and orientation in the environment. We formally define
the pose of a LiDAR sensor as the position and orientation of a frame attached to it with respect to a frame of reference. We represent the pose of a LiDAR by the rigid transformation mapping the reference frame to the LiDAR's frame.
A LiDAR scan at a pose $(R, \rho)$ is a finite set of pairs $(d, v)\in \mathbb{R}_{\ge0}\times \mathbb{R}^3$, where $d$ is a distance measurement and $v$ is a unit vector in the direction of the measurement. Given a LiDAR scan $S$ at a pose $(R, \rho)$, we can construct the corresponding point cloud $P\subset \mathbb{R}^3$:
$$
P = \left\{dv \in \mathbb{R}^3 | (d, v) \in S\right\}.
$$
It will be convenient to apply a rigid transformation $(R,\rho)$ not to a point but to a point cloud $P\subset \mathbb{R}^3$. We denote such application by $R(P)+\rho$ and define it as:
$$
R(P) + \rho = \left\{y\in \mathbb{R}^3\,\, \vert \,\,y=Rx+\rho,\quad x\in P\right\}.
$$

Consider two point clouds $P_1, P_2\subset \mathbb{R}^3$ generated from a LiDAR sensor at poses $p_1, p_2 \in SO(3) \times \mathbb{R}^3$, respectively, where $(R, \rho) \in SO(3) \times \mathbb{R}^3$ is the rigid transformation from $p_1$ to $p_2$. The objective of this paper is to solve the following problem:

\begin{center}
\textit{Given $P_1$ and $P_2$,  estimate $(R, \rho)$\\and provide worst-case error bounds on the estimate.
}
\end{center}

In the following section we describe the architecture and logic of our algorithm.

\section{PESTO Algorithm}
Given two point clouds, PESTO uses the Delaunay Tetrahedralization to efficiently sift through possible point-to-point correspondences from the first point cloud to the second. This search is conducted, not at the level of individual points, but at the level of tetrahedra for two reasons:
\begin{enumerate}
    \item A tetrahedron is defined by four points, which is the minimum number of points needed to uniquely define a rigid transformation between the point clouds;
    \item The tetrahedra constructed by the Delaunay Tetrahedralization only depend on neighboring points. This means that tetrahedra will not change when a point cloud is altered away from a tetrahedron.
\end{enumerate}
These observations suggest that tetrahedra represent local patterns and are a middle ground between fully processing the data to extract specific features and working directly with the points. Moreover, the ``features'' defined by tetrahedra are universal in the sense they do not depend on any knowledge about the environement where the measurements are taken.

Before describing the key steps in PESTO, we introduce some notation and concepts related to tetrahedra.

\subsection{Delaunay Tetrahedralization}

A tetrahedron $\Delta$ is a polyhedron in $\mathbb{R}^3$ consisting of four vertices and six edges with triangular facets. 
Given an ordering for the vertices, we denote the $i$th vertex of $\Delta$ by $v_i(\Delta)$. The length of the edge connecting the $i$th and $j$th vertices is given by $s_{ij}(\Delta)=s_{ji}(\Delta)=\Vert v_i(\Delta) - v_j(\Delta) \Vert_2$. For a tetrahedron with vertices $i,j,k,l$, we define the gap value for $v_i$ as $g_i(\Delta)=\min_{i\ne j\ne k}\{\vert s_{ij}(\Delta)-s_{ik}(\Delta) \vert\}$ and represents the smallest difference between the lengths of edges with $v_i$ as a vertex. We define the signed volume of a tetrahedron as $\mathrm{SVol}(\Delta)=\frac{1}{4}\text{det}\left(\begin{bmatrix}
    v_1(\Delta) - v_2(\Delta) & v_1(\Delta) - v_3(\Delta) & v_1(\Delta) - v_4(\Delta)
\end{bmatrix}\right)$. The centroid of $\Delta$ is defined as the first moment of the set consisting of the tetrahedron's vertices and is denoted by $\mu (\Delta)$.
A tetrahedron's circumsphere is a sphere whose boundary contains all the tetrahedron's vertices.

Given a set of points $Q\subset \mathbb{R}^3$ and a tetrahedron whose vertices are points of $Q$, the tetrahedron satisfies the empty circumsphere property if there are no points of $Q$ in the interior of the tetrahedron's circumsphere. We  utilize the Delaunay tetrahedralization, which is a generalization of the Delaunay triangulation for sets of three-dimensional points. The Delaunay tetrahedralization of a set of points $Q\subset \mathbb{R}^3$, denoted by $DT(Q)$, is a set of tetrahedra with points of $Q$ as vertices such that all tetrahedra in $DT(Q)$ satisfy the empty circumsphere property. 

\subsection{Key Steps in PESTO}
Starting from the Delaunay Tetrahedralization of the first point cloud, PESTO finds tetrahedra in the second that are similar to ones in the first point cloud to create pairs of matched tetrahedra. The pairs of similar tetrahedra are then used to compute the transformation that maximizes the common area between point clouds. 

\subsubsection{Step 1: Calculating Tetrahedron Vectors}
Given two point clouds $P_1$ and $P_2$, PESTO takes the Delaunay Tetrahedralization of each. To each tetrahedron $\Delta$, we assign a tetrahedron vector, defined as: 
\begin{multline*}
    \overrightarrow{\Delta} = \big(s_{12}(\Delta), s_{13}(\Delta),  s_{14}(\Delta), \\
     s_{23}(\Delta), s_{24}(\Delta), s_{34}(\Delta), \text{sign}\left(\mathrm{SVol}(\Delta)\right) \times \max_{i\ne j}s_{ij}(\Delta) \big).
\end{multline*}
where the vertices of $\Delta$ are ordered by gap value, i.e., for any vertices $i$ and $j$, $i<j$ implies $g_i(\Delta) \ge g_j(\Delta)$.
PESTO filters the tetrahedra to keep only those with gap sizes above a certain user-set threshold $\zeta\in \mathbb{R}_{> 0}$ and collects these in the sets $F(P_1)$ and $F(P_2)$

\subsubsection{Step 2: Matching by Tetrahedron Vectors}
From the pairs of tetrahedra $(\Delta_1, \Delta_2)\in F(P_1) \times F(P_2)$, PESTO collects the $m\in \mathbb{N}$ pairs that have the smallest distance between their tetrahedron vectors $\Vert \overrightarrow \Delta_1 - \overrightarrow \Delta_2 \Vert$ in the set $M_m(P_1, P_2)$. On each of the pairs in $M_m(P_1, P_2)$, PESTO uses Horn’s Method \cite{horn} to calculate the rigid transformation that aligns them.

\subsubsection{Step 3: Filter  Pairs by Heuristic}
For each pair of tetrahedra $(\Delta_1, \Delta_2)\in M_m(P_1, P_2)$ and its associated rigid transformation $(\Tilde R, \Tilde \rho)$, PESTO evaluates a heuristic on a small ball around each tetrahedron: $N_1=B_r(\mu(\Delta_1)) \cap P_1$,  $N_2=B_r(\mu(\Delta_2)) \cap P_2$, where $r\in \mathbb{R}_{>0}$ is chosen such that $\Delta_1 \subseteq  B_r(\mu(\Delta_1))$ and $\Delta_2 \subseteq  B_r(\mu(\Delta_2))$. PESTO evaluates the following heuristic function on $N_1$ and $N_2$:
\begin{multline}
H(\Delta_1, \Delta_2)=\left\Vert \mu(N_2) – (\Tilde R(\mu(N_1)) + \Tilde \rho) \right\Vert^2_2 + \\\nu \left\Vert \Sigma(N_2) - \Tilde R \Sigma(N_1) \Tilde R^\intercal \right\Vert_F  + \eta \left\Vert \overrightarrow{\Delta_1} - \overrightarrow{\Delta_2} \right\Vert_2,
\label{heuristic}
\end{multline}
where $\nu,\eta \in \mathbb{R}_{> 0}$ are design parameters, and collects the $k\in \mathbb{N}$ pairs with the smallest value of $H$ in the set $M_k^H(P_1, P_2)$.

\subsubsection{Step 4: Cost function}
For each pair of tetrahedra $(\Delta_1, \Delta_2)\in M_k^H(P_1, P_2)$ and its associated rigid transformation $(\Tilde R, \Tilde \rho)$, PESTO calculates its common area, defined as:
\begin{equation}
C\left(\Tilde R, \Tilde \rho\right) = \left\{\left(x_1, x_2\right) \in P_1 \times P_2 \vert \left\Vert \Tilde Rx_1+ \Tilde \rho - x_2\right\Vert_2 \le \xi\right\},
\end{equation}
where $\xi \in \mathbb{R}_{>0}$ is a user-defined parameter.
PESTO uses the following cost function to evaluate the size of the common area while also checking how well the points of the common area are aligned:
\begin{equation}
     J(\Tilde R, \Tilde \rho) = \text{Vol } \partial C\left(\Tilde R, \Tilde \rho\right)-\beta \sum_{(x_1,x_2) \in C\left(\Tilde R, \Tilde \rho\right) } \left\Vert \Tilde Rx_1 +\Tilde \rho - x_2 \right\Vert_2,
\end{equation}
where $\beta\in \mathbb{R}_{>0}$ is a user-defined parameter. A description on how we calculate $\text{Vol } \partial C$ can be found in Appendix A. Finally, PESTO returns the rigid transformation that maximizes the cost function $J$.

\subsection{Discussion on Algorithm}
We chose the form of the tetrahedron vector to make it invariant under rigid transformations; the signed volume term ensures that two tetrahedra related by a reflection have a last term that differs by sign. Additionally, filtering based on gap sizes renders the algorithm robust to noise by removing tetrahedra whose side ordering may change by a small amount of noise and reducing the number of erroneous matches due to sensor noise.

The heuristic (\ref{heuristic}) is based on necessary conditions on the first and second moments of two point clouds related by a rigid transformation, derived in our previous work on PASTA \cite{pasta}. We filter tetrahedra in two steps because calculating the distance between tetrahedron vectors in Step 2 is more computationally efficient than calculating the heuristic in Step 3. The asymptotic behavior of the algorithm's runtime as the size of the point clouds increases is dominated by the construction of the Delaunay Tetrahedralization, which for point clouds comprised of $n \in \mathbb{N}$ points has a computational complexity of $O(n\log n)$; PESTO therefore also has a computational complexity of $O(n\log n)$ for point clouds of size $n$. \footnote[3]{Using the KDTree data structure, searching for corresponding points to calculate the heuristic function and the common area scales at a rate of $O(n\log n)$.} The runtime of the algorithm can be made faster by downsampling the point clouds to use fewer points, which is already common practice when running point cloud registration algorithms in real-time applications, or by decreasing $m$ and $k$ to check fewer tetrahedron pairs.

\section{Worst-Case Error Bounds}
In addition to its rigid transformation estimate, PESTO provides theoretical worst-case guarantees on the error between its estimate and the true transformation. Proofs for Theorems 1 and 2 can be found in Appendix B.

PESTO's bounds are calculated using four point-to-point correspondences from the vertices of the matched tetrahedron pair used to calculate its rigid transformation estimate.
In the ideal case, the vertices of the matched tetrahedra provide a set of four point-to-point correspondences $\left\{(x_1,y_1), (x_2, y_2), (x_3, y_3), (x_4, y_4)\right\}$,  which are related exactly by the true rigid transformation $(R, \rho)$, i.e., $y_i = R x_i + \rho$.
When working with the real point clouds, the points $y_i$ are perturbed by ``noise'' $\epsilon_i \in B_\delta(0)$ and the correspondences become $\left\{(x_1, \hat y_1), (x_2, \hat y_2), (x_3, \hat y_3), (x_4, \hat y_4)\right\}$, where $\hat y_i = y_i+ \epsilon_i$. Note that the ``noise'' term $\varepsilon_i$ accounts for measurement noise as well as the quantization effects created by only being able to measure distances in a finite number of directions radiating out from the LiDAR sensor.

Theoretical guarantees on the quality of PESTO's rigid transformation estimate are based on the following assumption.

\textbf{Assumption 1:}
Let $(x_i,\hat{y}_i)$, $i=1,\hdots,4$, be the correspondences used by PESTO to compute its output. The ``noise'' terms $\varepsilon_i=\hat{y}_i-Rx_i-\rho$, where $(R,\rho)$ is the rigid transformation between the poses of the LiDAR sensor when the measurements were acquired, satisfy $\Vert \varepsilon_i\Vert\le \delta(\Vert x_i\Vert )$ for all $i=1,\hdots,4$ where $\delta: \mathbb{R}_{\ge 0} \rightarrow \mathbb{R}_{\ge 0}$ is a non-negative, non-decreasing function of the magnitude of the distance measurements.

Within Assumption 1, we encapsulate two concepts: first, that the magnitude of the ``noise" on LiDAR measurements is bounded but may increase with the distance, and second, that maximizing the common area will result in approximately correct correspondences. In practice, the first can be easily satisfied given knowledge of the LiDAR sensor characteristics and the environment in which it is deployed. The second relates to the size of the common area that will be found during a mission. The user needs to use their domain knowledge to determine if this is a reasonable assumption.

To calculate the rotation bound, we first define:
$$x_i'=x_i-\mu\left(\{x_1,\hdots,x_4\}\right), \quad y_i'=y_i-\mu\left(\{y_1,\hdots,y_4\}\right),$$
$$\hat y_i'=\hat y_i-\mu\left(\{\hat y_1,\hdots,\hat y_4\}\right),$$
and arrange these points into matrices $A=[x_1'|x_2'|x_3'|x_4']$, $B=[y_1'|y_2'|y_3'|y_4']$ , $\Delta B=[\hat y_1'|\hat y_2'|\hat y_3'|\hat y_4']-B$. Using these matrices we can upper bound the error in the rotation matrix as follows.

\begin{theorem}
\label{thm:rot}
Denote the singular values of the matrix $A$ by $\sigma_1 \ge \sigma_2 \ge \sigma_3 > 0$ and let  $\gamma\in \mathbb{R}_{>0}$ satisfy: 
\begin{equation}
\min_{R \in SO(3)} \left\Vert RA - B \right\Vert_F\le \gamma < \left(\sigma_{2}^2 +\sigma_3^2\right)^{\frac{1}{2}}-2\omega,
\end{equation}
where $\omega=\sqrt{\sum_{i=1}^4 \delta(\Vert x_i\Vert )^2 + \frac{3}{4}\left(\sum_{i=1}^4\delta(\Vert x_i\Vert )\right)^2}$. Under Assumption 1, the error in the rotation matrix produced by PESTO is upper bounded as follows: 
\begin{equation}
 \left\Vert R - \hat{R}\right\Vert_F\le 2 \sqrt{1-\left(1-\frac{\omega^2}{ \left( \left(\sigma_{2}^2+\sigma_3^2 \right)^{\frac{1}{2}} -\gamma\right)^2} \right)^\frac{1}{2}}.
 \end{equation}
\end{theorem}

To provide a bound for the translation error, we introduce the alignment error $\theta_i =\hat{y}_i-\hat R x_i - \hat \rho$ for the pairs $(x_i,\hat{y}_i)$, $i=1,\hdots,4$ and define:
$$w_{ij} = y_i - y_j, \quad\hat w_{ij} = \hat y_i - \hat y_j, \quad z_{ij} = x_i –x_j. \quad \forall i<  j$$

\begin{theorem}
\label{thm:trans}
Under Assumption 1, the error in the translation vector produced by PESTO is upper bounded as follows:  
    \begin{multline}
        \Vert \rho - \hat \rho\Vert_2 \le \frac{1}{4}\left(\sum_{i=1}^4\delta(\Vert x_i \Vert ) \right.\\ \left.+ \sum_{j<k} \left\vert \alpha_{jk}^{(i)}\right\vert \left(\delta(\Vert x_j\Vert) + \delta(\Vert x_k\Vert ) + \left\Vert \theta_j \right\Vert_2 + \left\Vert \theta_k\right\Vert_2 \right)\right), 
    \end{multline}
    where $\alpha^{(i)}_{jk} \in \mathbb{R}$ are selected so that $x_i = \sum_{j<k} \alpha_{jk}^{(i)} z_{jk}$.
\end{theorem}

To get a tighter translation bound, we use the point-to-point correspondences collected when computing $C(\hat R, \hat \rho)$ to create a basis of vectors $z_{jk}$ that reduce the magnitude of the coefficients $\alpha^{(i)}_{jk}$.

\section{Experiments}


We first compare PESTO’s registration performance against other point cloud registration algorithms on a real-world dataset, the KITTI Odometry Benchmark~\cite{KITTI}, on which 
we also experimentally validate PESTO's worst-case error bounds. We evaluate the relative rotation error (RRE) and relative translation error (RTE), defined as follows:
$$RTE(\rho_1, \rho_2) = \left\Vert \rho_1 - \rho_2 \right\Vert_2,$$
$$RRE(R_1, R_2) = \arccos \frac{\text{tr}\left(R_1^\intercal R_2\right)-1}{2}.$$
We implemented PESTO in Python 3, relying on SciPy's implementation of the Delaunay Tetrahedralization algorithm. For all experiments, the parameters for PESTO were set as follows: $\zeta=0.005$, $\nu=50$, $\eta=10$, $\beta=0.0001$, $\xi=0.02$, $\gamma=0.018$, $m=30$, and $k=5$. We scale the point clouds to be within the interval $[-1,1]^3$. We compare PESTO against other model-based algorithms used for point cloud alignment in LiDAR localization applications: ICP, RANSAC, and TEASER. 
TEASER cannot run with large amounts of points due to memory constraints and requires correspondences to be known ahead of time, so we have used Fast Point Feature Histograms (FPFH) to determine correspondences; RANSAC is also fed correspondences from FPFH for better performance and faster runtime. For the best performance, we did not set a runtime cap on RANSAC and ICP. As the other tools do not exploit any knowledge about the motion of the vehicle, ICP is provided a random transformation as an initial guess, with a rotation sampled uniformly from $SO(3)$ and a translation uniformly sampled from $[-1,1]^3$.

\begin{figure}
    \centering
    \includegraphics[width=\linewidth]{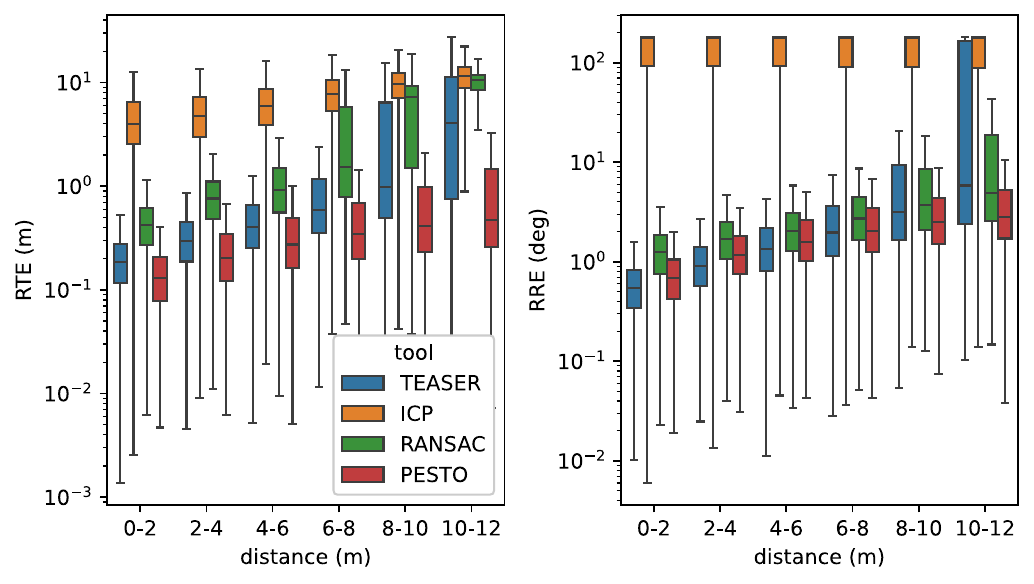}
    \caption{A comparison of PESTO's performance on the KITTI Odometry Benchmark.}
    \label{fig:kitti}
\end{figure}

For each LiDAR scan, we downsample each point cloud uniformly by keeping every tenth point, creating a point cloud comprised of around 10,000 points. Feature extraction further reduces the number of points to less than 1,000 pairs of corresponding points used to generate estimates with RANSAC and TEASER. A pair of scans are correctly aligned if $RRE < 5^\circ$ and $RTE < \frac{1}{2}+\frac{1}{6}\left\lfloor \frac{\left\Vert \rho \right\Vert}{2} \right\rfloor$ meters, where $\rho$ is the true translation between the scans and $\left\lfloor \frac{\left\Vert \rho \right\Vert}{2} \right\rfloor$ is the floor of $\frac{\left\Vert \rho \right\Vert}{2}$. Note that the threshold for $RTE$ grows with distance. We measure recall as the percent of correctly aligned scans. 

For each of the 11 sequences, the tools align each scan with each of the 15 subsequent scans (recall that scans are taken along the path followed by the LiDAR equipped vehicle). We sorted the pairs of scans aligned into bins based on how far the vehicle moved between scans, using this as an approximate measure of the amount of overlap between the scans.

Results from the experiments on the KITTI Odometry Benchmark can be found in Table \ref{table:recall} and in Figure \ref{fig:kitti}. PESTO outperforms the other tools in registration recall at every distance between scans, and its performance degrades more slowly than the other tools as the distance between scans increases. TEASER matches PESTO’s performance at smaller distances but correctly aligns only half as many scans as PESTO when the distance between scans is between 10 and 12 meters. RANSAC, which relies on features from FPFH like TEASER, has a similar drop in performance as TEASER but performs worse than TEASER at every distance. Similarly to its performance on the Stanford models, ICP is largely unable to find the correct alignment.

On the KITTI dataset, PESTO had an average runtime of $0.32$ seconds, compared to $0.07$ seconds for TEASER, $0.05$ seconds for RANSAC, and $0.83$ seconds for ICP. While PESTO has a longer runtime than TEASER and RANSAC, we note that PESTO is using more than 10 times the number of points and that we have not yet optimized our algorithm for speed as the other algorithms presumably have been; faster runtimes for PESTO can be achieved either by further downsampling the point clouds or by using features at the cost of a slight decrease in performance. Furthermore, we believe that PESTO might able to run at a lower frequency given its better recall at larger distances.

\begin{table}[H]
\caption{Registration recall on the KITTI Odometry Benchmark.}
\label{table:recall}
\begin{tabular}{l|llll}
Distance (m) & TEASER & RANSAC & ICP   & PESTO \\ \hline
0-2          & 0.973  & 0.731  & 0.134 & \textbf{0.976} \\
2-4          & 0.928  & 0.499  & 0.116 & \textbf{0.936} \\
4-6          & 0.850  & 0.496  & 0.069 & \textbf{0.878} \\
6-8          & 0.708  & 0.354  & 0.034 & \textbf{0.825} \\
8-10         & 0.528  & 0.205  & 0.011 & \textbf{0.755} \\
10-12        & 0.352  & 0.093 & 0.002 & \textbf{0.692}
\end{tabular}

\end{table}

We validated PESTO’s performance bounds on a sequence of 100 scans from the KITTI Odometry Benchmark sequence 00, as shown in Figure \ref{fig:bounds}. Each scan is aligned with the third subsequent scan. To calculate the bounds, we modeled the noise on the LiDAR measurements as an affine function of the magnitude of the distance measurement $\delta(r)=0.1r + 0.02$. We converted the norm of the difference between the true and estimated rotation matrices into degrees for ease of interpretation. The error bounds are relatively conservative but are tighter than the bounds in the Stanford Bunny experiment because the tetrahedra used to align the point clouds in KITTI have larger singular values, which tightens the bound. The spikes in the bound and error correspond to when the tetrahedra used to calculate the estimate have less similar tetrahedra vectors.

\begin{figure}
    \centering
    \includegraphics[width=\columnwidth]{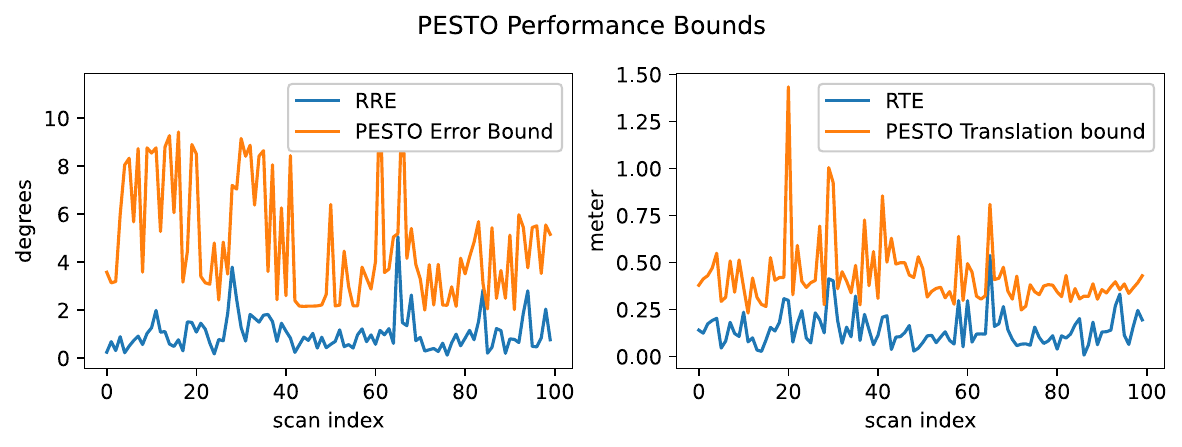}
    \caption{PESTO's bounds as evaluated on a sequence of 100 scans from the KITTI Odometry Benchmark.}
    \label{fig:bounds}
\end{figure}

\subsection{Comparison to Neural Network Registration}
We compared PESTO to GeoTransformer \cite{geotransformer}, a neural network point cloud registration tool designed to align point clouds with limited overlap, on the KITTI 10m Benchmark used to evaluate GeoTransformer. For this experiment, we defined a correct rigid transformation as incurring an error within 10 degrees and 3 meters. GeoTransformer outperforms PESTO with a registration recall of 99\% to PESTO's 85\% and has an average runtime of $0.35$ seconds. However, GeoTransformer sometimes provides completely wrong alignments even for pairs of point clouds that PESTO aligns correctly. In Figure \ref{fig:neural_comparison}, we compare the performance of PESTO and GeoTransformer on a sequence of scans from KITTI sequence 7, starting at scan 2 and aligning each subsequent scan with the first scan and stopping when the vehicle has moved 15 meters from its starting position. For most scans, both tools find a correct rigid transformation, but GeoTransformer has slightly better performance, providing rigid transformation estimates with smaller error. However, towards the end of the sequence, GeoTransformer returns rigid transformations that are wildly incorrect; a visualization of the estimated rigid transformation in one of these cases can be seen in Figure \ref{fig:kitti_example}. Given that PESTO correctly aligned the point clouds, it appears that there are enough corresponding features between the point clouds to find a correct rigid transformation, and considering GeoTransformer is trained specifically on a dataset comprised of point clouds at least 10 meters apart, there is no discernible reason why GeoTransformer could not align the point clouds. PESTO, with its error bounds and model-based algorithm, is more interpretable and arguably a better choice for safety-critical applications. With its worst-case error bounds, PESTO can supervise GeoTransformer's rigid transformation estimates, using its worst-case error bounds to discard estimates that lie outside its error bounds and replace them with a transformation within the error bounds. 

\begin{figure}
    \centering
    \includegraphics[width=\linewidth]{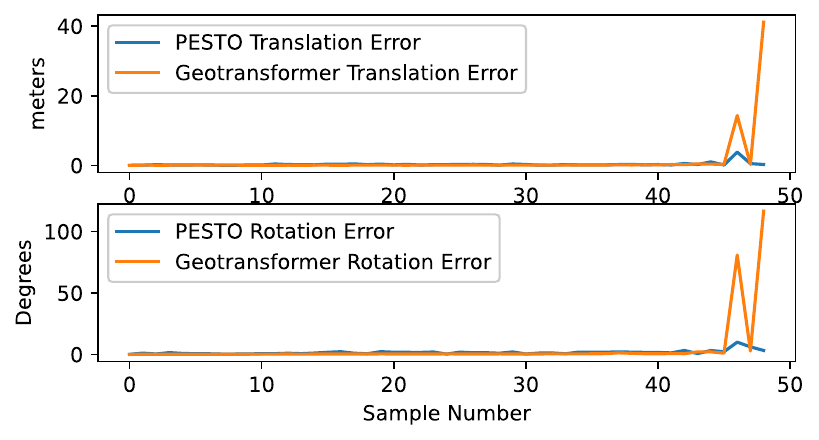}
    \caption{A comparison of alignment error from PESTO and neural network GeoTransformer aligning a sequence of point clouds from KITTI}
    \label{fig:neural_comparison}
\end{figure}

\section{Conclusion and Future Work}
We propose PESTO, a LiDAR point cloud registration algorithm with formal worst-case error bounds, and show that PESTO performs well on real-world point cloud registration problems, especially as the overlap between the point clouds shrinks. PESTO's worst-case error bounds provide invaluable information on the uncertainty in its estimate, making PESTO's estimates more interpretable and useful in real-world, safety-critical applications; future work could leverage these error bounds to supervise neural network registration tools, combining the performance of neural networks with the formal guarantees provided by PESTO. We are currently working to implement PESTO in C++ to reduce its running time for use on larger datasets under real-time constraints.





\section*{APPENDIX}

\subsection{Calculating $Vol(\partial C)$}

We will use the notation $C_1(R, \rho)$ to denote the projection of $C$ in $P_1$:

\begin{equation*}
    C_1(R,\rho) =\left\{x_1 \in P_1\vert (x_1, x_2) \in C(R, \rho)\right\}.
\end{equation*}
We define a triangle as a polygon with three vertices and three edges and define the interior of a triangle $\tau$ with vertices $(v_1, v_2, v_3)$ as the set:
$$I(\tau) =\left\{x \in \mathbb{R}^3 \vert x = \sum_{i=1}^3 \alpha_i v_i, 0<\alpha_i <1, \sum_{i=1}^3 \alpha_i = 1\right\}.$$

We will approximate $Vol\left(\partial C\right)$ as the surface area of $\partial C$. First, we assign an ordering to the points in $C_1\left(R,\rho\right)$ and denote the $i^{th}$ element of $C_1\left(R,\rho\right)$ by $C_1(i)$. We define points $C_1(i)$ and $C_1(j)$ as neighbors if the distance
$$\left\Vert \frac{C_1(i)}{\left\Vert C_1(i) \right\Vert_2} - \frac{C_1(j)}{\left\Vert C_1(j) \right\Vert_2}\right\Vert_2$$
is below threshold $\gamma\in \mathbb{R}_{>0}$, where $\gamma$ is a design parameter. 
To approximate the volume of the boundary of $C_1\left(R, \rho\right)$, we triangulate $C_1(R, \rho)$ by creating triangles out of neighboring points and calculate the surface area using the sum of their areas. We construct the set $\Phi(C_1)$ of triangles with the following properties:
\begin{enumerate}
    \item For any triangle in $\Phi(C_1)$, each of its vertices are neighbors of the other two vertices;
    \item For any two triangles in $\Phi(C_1)$, the intersection of their interiors is empty;
    \item For any point $C_1(i)$, the ray from the origin to through $C_1(i)$ does not intersect the interior of any triangle in $\Phi$, i.e., there does not exist constant $\omega\in \mathbb{R}_{>0}$ such that $\omega C_1(i) \in I(\tau), \forall \tau\in \Phi(C_1)$;
    \item No edges can be added between neighboring points $C_1(i)$ and $C_1(j)$ without intersecting another edge.
\end{enumerate}

$\text{Vol } \partial C$ is approximated as the combined area of all the triangles in $\Phi(C_1)$:
\begin{equation}
    \text{Vol } \partial C_1\left(R, \rho\right)=\sum_{\tau\in \Phi(C_1)} Area(\tau),
\end{equation}
where $Area(\tau)$ is the area of the triangle $\tau$.

\subsection{Proofs for Formal Worst-Case Guarantees}

\textit{Proof of Theorem \ref{thm:rot}:}

Note that we can recover $R$ and $\hat R$ using $A$, $B$, and $\Delta B$ as follows:
$$ R = \text{arg } \min_{R \in SO(3)} \left\Vert RA - B\right\Vert_F,$$
$$\hat R = \text{arg } \min_{R \in SO(3)} \left\Vert RA - (B+\Delta B)\right\Vert_F,$$
where Horn's method provides the solution to the above minimization problems~\cite{horn}. 

We can bound the difference between $y_i'$ and $\hat y_i'$ as follows:
\begin{align*}
    \left\Vert y_i' -\hat y_i'\right\Vert_2  &= \left\Vert y_i - \frac{1}{4}\sum_{j=1}^4y_j -\hat y_i +\frac{1}{4}\sum_{j=1}^4  \hat y_j \right\Vert_2 \\
    & = \left\Vert -\varepsilon_i + \frac{1}{4} \sum_{i=1}^4\epsilon_j \right\Vert_2 \\
    & \le \left \Vert \varepsilon_i\right\Vert_2 + \frac{1}{4}\sum_{j=1}^4\left\Vert \varepsilon_j \right\Vert_2 \\
    &\le \delta(\Vert x_i\Vert) + \frac{1}{4}\sum^4_{j=1}\delta(\Vert x_j\Vert).
\end{align*}

Using the above expression, we can bound the perturbation $\Delta B$ as follows:

\begin{align*}
    \left\Vert \Delta B \right\Vert_F &= \sqrt{\sum_{i=1}^4 (y_i'-\hat y_i')^\intercal (y_i'-\hat y_i)} \\
    &\le\sqrt{\sum_{i=1}^4 \left(\delta(\Vert x_i\Vert) + \frac{1}{4}\sum_{j=1}^4\delta(\Vert x_j\Vert)\right)^2} \\
    &\le\sqrt{\sum_{i=1}^4 \delta(\Vert x_i\Vert)^2 + \frac{3}{4}\left(\sum_{i=1}\delta(\Vert x_i\Vert)\right)^2} = \omega .
\end{align*}

A bound on the rotation error, $\Vert R-\hat R\Vert_F$ is now directly obtained from Theorem 2.1 page 2 in~\cite{Söderkvist1993} where we used Assumption 1 to replace $\varepsilon_B$ in \cite{Söderkvist1993} with $4\delta$. 
\QED

\textit{Proof of Theorem \ref{thm:trans}:}

    Using all the correspondences $(x_i, y_i)$, we can express $\rho$ as follows: 
    \begin{equation}
        \rho = \frac{1}{4}\sum_{i=1}^4\left( y_i - Rx_i\right).
        \label{ideal}
    \end{equation}
    
    Using Horn's method, we calculate $\hat \rho$ using the following expression:
    \begin{equation}
        \hat \rho = \frac{1}{4}\sum_{i=1}^4\left( \hat y_i - \hat Rx_i\right).
        \label{horn}
    \end{equation}
    
    Subtracting \eqref{ideal} from \eqref{horn}, we obtain: 
    \begin{align*}
        \left\Vert \hat \rho - \rho \right\Vert_2  &= \frac{1}{4} \left\Vert   \sum_{i=1}^4 \hat y_i - y_i - \hat Rx_i + Rx_i\right\Vert_2 \\
        & = \frac{1}{4}\left \Vert \sum_{i=1}^4 \epsilon_i + (R-\hat R) x_i\right\Vert_2  \\
        & \le \frac{1}{4} \sum_{i=1}^4 \left\Vert \epsilon_i\right\Vert_2  + \sum_{i=1}^4 \left\Vert (R-\hat R) x_i\right\Vert_2  \\
        & \le \frac{1}{4}\sum_{i=1}\delta(x_i) + \frac{1}{4}\sum _{i=1}^4\left\Vert (R-\hat R)x_i\right\Vert_2 .
    \end{align*}

    To bound $\left\Vert \left(R-\hat R\right) x_i\right\Vert$, we can rewrite each $x_i$ as a linear combination of vectors $z_{jk}$:
    \begin{align*}
        \left\Vert\sum _{i=1}^4\left(R-\hat R\right)x_i\right\Vert_2 &= \left\Vert \sum _{i=1}^4 \left(R-\hat R\right)\left(\sum _{j < k}\alpha_{jk}^{(i)} z_{jk}\right)\right\Vert_2 \\
        &\le \sum_{i=1}^4\sum_{j<k} \left\vert \alpha_{jk}^{(i)} \right\vert \left\Vert \left(R-\hat R\right) z_{jk} \right\Vert_2.
    \end{align*}

    To bound $\left\Vert \left(R-\hat R\right)z_{jk}\right\Vert_2$, first note that $w_{ij}$ and $z_{ij}$ are related by:
    \begin{equation}
       w_{ij} = y_i - y_j = Rx_i + \rho - Rx_j - \rho = Rz_{ij},
       \label{ideal_z_w}
    \end{equation}
    and similarly relate $\hat w_{ij}$ and $z_{ij}$ by:
    \begin{align}
       \hat w_{ij} &= \hat y_i - \hat y_j   \notag \\
       &= \hat Rx_i +\hat \rho + \theta_i - \hat Rx_j - \hat \rho - \theta_j  \notag \\
       &=\hat Rz_{ij} + \theta_i - \theta_j. \label{z_w}
    \end{align}
    Using \eqref{z_w} and \eqref{ideal_z_w}, we get the following expression bounding $\left\Vert (R-\hat R)z_{jk}\right\Vert_2$:
    \begin{align*}
        \left\Vert R z_{jk}- \hat R z_{jk} \right\Vert_2 &= \left\Vert w_{jk} - w_{jk} + \epsilon_j- \epsilon_k - \theta_j + \theta_k \right\Vert_2 \\
        & = \left\Vert \epsilon_j + \epsilon_k - \theta_j + \theta_k \right\Vert_2 \\
        &\le \left\Vert \epsilon_j \right\Vert_2  + \left\Vert \epsilon_k\right\Vert_2 + \left\Vert \theta_j\right\Vert_2 + \left\Vert \theta_k\right\Vert_2  \\
        &\le  \delta(x_i) + \delta(x_j) + \left\Vert \theta_j \right\Vert_2 + \left\Vert \theta_k\right\Vert_2 .
    \end{align*}

    Combining these together, we get the final bound:
    \begin{align*}
        \left\Vert \hat \rho - \rho \right\Vert_2  &\le \frac{1}{4}\sum\delta(\Vert x_i\Vert ) + \frac{1}{4}\left\Vert\sum _{i=1}^4\left(R-\hat R\right)x_i\right\Vert_2 \\
        &\le \frac{1}{4}\sum_{i=1}\delta(\Vert x_i\Vert ) + \frac{1}{4}\sum_{i=1}^4 \sum_{j<k} \left\vert \alpha_{jk}^{(i)}\right\vert \left\Vert\left(R- \hat R\right) z_{jk} \right\Vert_2  \\
        &\le \frac{1}{4}\sum_{i=1}^4\delta(\Vert x_i\Vert ) +\sum_{i=1}^4 \sum_{j<k} \left\vert \alpha_{jk}^{(i)}\right\vert \left(\delta(\Vert x_j\Vert) \right. \\ 
        &  \qquad + \left. \delta(\Vert x_k\Vert ) + \left\Vert \theta_j \right\Vert_2 + \left\Vert \theta_k\right\Vert_2 \right).
    \end{align*}

\QED


\bibliographystyle{IEEEtran} 
\bibliography{bib}

\end{document}